\documentclass[11pt]{article}

\usepackage[T1]{fontenc}
\usepackage[utf8]{inputenc}
\usepackage{microtype}
\usepackage{geometry}
\usepackage{amsmath,amssymb,amsthm,amsfonts}
\usepackage{mathtools}
\usepackage{graphicx}
\usepackage{xcolor}

\usepackage{subcaption}
\graphicspath{{./}{./figs/}{./Figures/}}
\DeclareGraphicsExtensions{.pdf,.png,.jpg}

\usepackage{tikz}
\usetikzlibrary{arrows.meta,positioning,calc,fit}

\usepackage{xparse}

\NewDocumentCommand{\Id}{o}{\mathbf{I}\IfNoValueTF{#1}{}{_{#1}}}
\theoremstyle{plain}
\newtheorem{theorem}{Theorem}
\newtheorem{corollary}{Corollary}
\theoremstyle{remark}
\newtheorem{remark}{Remark}
\newcommand{\one}{\mathbf 1}

\usepackage{hyperref}

\title{Steady-State Spread Bounds for Graph Diffusion via \\Laplacian Regularisation in Networked Systems}
\author{Ardavan Rahimian\\School of Engineering, Ulster University\\Belfast, BT15 1AP, UK\\\texttt{a.rahimian@ulster.ac.uk}}
\date{}

\begin{document}
\maketitle

\begin{abstract}
We study how far a diffusion process on a graph can deviate from a designed starting pattern when the pattern is generated via Laplacian regularisation. Under standard stability conditions for undirected, entrywise nonnegative graphs, we give a closed-form, instance-specific upper bound on the steady-state spread, measured as the relative change between the final and initial profiles. The bound separates two effects: (i) an irreducible term determined by the graph's maximum node degree, and (ii) a design-controlled term that shrinks as the regularisation strength increases (with an inverse square-root law). This leads to a design rule: given any target limit on spread, one can choose a sufficient regularisation strength in closed form. Although one motivating application is array beamforming---where the initial pattern is the squared magnitude of the beamformer weights---the result applies to any scenario that first enforces Laplacian smoothness and then evolves by linear diffusion on a graph. Overall, the guarantee is non-asymptotic, easy to compute, and certifies the maximum steady-state deviation.
\end{abstract}

\section{Introduction}
In networked sensing and wireless systems, power or information seeded at a few elements rarely remains localised: coupling through the physical or logical graph causes \emph{spreading} across neighbours. Examples include interference leakage in arrays and cooperative sensor networks, energy diffusion in reconfigurable surfaces~\cite{Wu2021IRS}, and activation smoothing in graph neural networks~\cite{Gama2020Graphs,Dong2020GSPMachineLearning}. In many of these settings, we can aggressively \emph{shape the initial profile} (e.g., by regularised design~\cite{Huang2023Robust}), but have limited control over the subsequent network dynamics driven by a fixed adjacency. This paper studies a fundamental question: \emph{how much spreading can occur at steady state if the initial profile is engineered via graph-Laplacian regularisation?}

We model network dynamics with the linear recursion
\(
\mathbf p_{t+1}=\rho\,\mathbf G\,\mathbf p_t+(1-\rho)\,\mathbf p_0
\),
where $\mathbf G$ is a symmetric, entrywise nonnegative adjacency and $\rho\in(0,1)$ is a propagation factor; the initial profile $\mathbf p_0=|\mathbf w^\star|^2$ is produced by a Laplacian-regularised design that balances a positive semidefinite data-fit term $\mathbf R_{\mathrm{in}}$ (e.g., interference-plus-noise in array processing) against a graph-smoothness penalty $(\mu\,\mathbf w^H\mathbf L\mathbf w)$. We prove a closed-form, instance-wise upper bound on \emph{steady-state spreading}
\(
\xi=\|\mathbf p_\infty-\mathbf p_0\|_2/\|\mathbf p_0\|_2
\)
that cleanly separates: (i) an \emph{irreducible} degree-driven term $d_{\max}$; and (ii) a \emph{design-controlled} term that decays as $O(\mu^{-1/2})$. This yields a closed-form recipe for selecting $\mu$ from a target budget $\xi_{\rm target}$.

Unlike isotropic shrinkage (ridge, $\mathbf L=\Id$) or using diagonal loading in $\mathbf R_{\rm in}$ alone, Laplacian penalties exploit the graph's geometry~\cite{Shuman2013Emerging,Sandryhaila2013Discrete,Ortega2018GSP}: neighbouring weights are encouraged to be similar, reducing high-frequency graph energy.

Classical array processing provides beamformers (e.g., MVDR) and regularised variants; robust designs focus on mismatch and sample scarcity~\cite{Huang2023Robust}. Yet, to our knowledge, no \emph{closed-form, instance-wise} bound ties a graph-regularised beamformer to the \emph{eventual} diffusion of its squared magnitude under a network recursion. Bounds in wireless often target SINR, outage, or interference temperature under channel models~\cite{NaderiAlizadeh2023Learning,Wang2024ENGNN}; diffusion on graphs is analysed via mixing times or spectral gaps~\cite{Chung1997Spectral,Hammond2011Wavelets}. Our result bridges these viewpoints: it quantifies steady-state diffusion \emph{of an engineered initial pattern} with an explicit $1/\sqrt{\mu}$ law, exposing controllable levers ($\mu$) and structural ones ($\rho$, $d_{\max}$, $\|\mathbf G\|_2$, $\lambda_{\max}(\mathbf L)$).

The $O(\mu^{-1/2})$ decay makes the design trade-off explicit: more smoothing means provably less spreading (and we quantify ``how much less''). The corollary maps a target budget $\xi_{\rm target}$ to a closed-form $\mu$. We show how this behaves with \(\rho\), including the caveat that the prefactor $C(\rho,\mathbf G)=\frac{(1-\rho)\rho}{1-\rho\|\mathbf G\|_2}$ grows near the stability boundary $\rho\|\mathbf G\|_2\uparrow 1$.

\textbf{Contributions.}
\emph{(1)} A closed-form upper bound on steady-state spreading with explicit $O(\mu^{-1/2})$ decay and a split between degree coupling and penalisation-controlled terms;
\emph{(2)} A computable corollary converting a spreading budget into a choice of~$\mu$;
\emph{(3)} A proof method combining steady-state expansion, a Laplacian energy argument tailored to $\mathbf p_0=|\mathbf w^\star|^2$, and a KKT reference-vector bound.

\paragraph{Scope and caveats.}
We consider undirected, entrywise nonnegative graphs, linear time-invariant diffusion, and place the Laplacian penalty outside $\mathbf R_{\mathrm{in}}$ (crucial for the $1/\sqrt{\mu}$ law). Prefactors are instance-specific and the bound is conservative; see Sec.~\ref{sec:limitations} for a detailed discussion.

\paragraph{Interpretation in applications.}
Although our analysis is abstract, the variables admit concrete interpretations in several domains. In an array or intelligent surface, graph nodes represent elements, $\mathbf w$ collects complex beamformer weights, and $\mathbf p_0 = |\mathbf w^\star|^2$ describes the per-element power profile that later diffuses across coupled elements according to~\eqref{eq:model-diffusion}. In wireless or networked sensing settings, nodes represent devices or base stations, $\mathbf p_t$ captures per-node power, load, or activation, and the graph $\mathbf G$ encodes communication or interference links. More broadly, in graph signal processing or GNNs, $\mathbf w$ can be viewed as design parameters or an activation pattern that is shaped by Laplacian regularisation, after which the same graph governs subsequent diffusion or message-passing dynamics. Throughout, our framework is intended as a general robust-design tool that can be instantiated in these concrete settings, rather than a literal physical model tied to a single propagation law.

\section{Related Work}\label{sec:related}

\paragraph{Beamforming and regularisation.}
Classical constrained designs (e.g., MVDR/Capon) minimise interference-plus-noise under a distortionless constraint. Robust variants manage steering uncertainty and sample scarcity via diagonal loading or shrinkage of covariance estimates~\cite{Huang2023Robust}. More recent approaches add explicit regularisers—sparsity, group penalties, and smoothness—to encode structure in the array or network. Our formulation uses a \emph{graph Laplacian} penalty, which—unlike isotropic ridge—suppresses high-frequency graph components and exploits topology~\cite{Shuman2013Emerging,Leus2023GSPHistory}.

\paragraph{Graph signal processing (GSP) for wireless systems.}
GSP has been used to regularise estimation, denoise sensor measurements, and constrain control actions on networks~\cite{Ortega2018GSP,Ruiz2021Graphon}. Laplacian smoothness priors are common to promote spatial coherence across sensors or antennas~\cite{Feng2022DataReconstruction}. Graph-based architectures have been applied to wireless resource management~\cite{Wang2024ENGNN,Yang2017InterferenceGraph} and machine learning (ML) on network data~\cite{Dong2020GSPMachineLearning,Gama2019ConvolutionalNN}. In contrast to typical GSP regressions, we tie the \emph{engineered} graph-smooth profile to a subsequent diffusion process and provide an instance-wise steady-state bound.

\paragraph{Diffusion on graphs and bounds.}
Bounds in network diffusion target mixing times or spectral decay and assume arbitrary initial conditions~\cite{Chung1997Spectral,Hammond2011Wavelets}. Heat diffusion on graphs has been studied for learning graph topology~\cite{Thanou2017HeatDiffusion}, while distributed consensus under imperfect communication~\cite{Kar2009Distributed} and privacy-preserving optimization~\cite{Yang2025DifferentiallyPrivate} examine related dynamics. Graph neural networks for wireless systems~\cite{NaderiAlizadeh2023Learning} and general GNN architectures~\cite{Gama2020Graphs} leverage graph structure but do not provide worst-case spreading bounds for Laplacian-regularised initial conditions. Our focus differs: we analyse the steady-state of a \emph{linear} diffusion driven by an initial condition shaped by a \emph{specific Laplacian-regularised optimisation}, and we obtain a closed form that separates a degree term ($d_{\max}$) from a design term scaling as $1/\sqrt{\mu}$.

\paragraph{Design rules and performance guarantees.}
Closed-form tuning rules are prized in wireless for their operational simplicity (e.g., loading rules, power budgets). Our corollary delivers such a rule for the regularisation strength~$\mu$, with an explicit feasibility floor \(C(\rho,\mathbf G)\,d_{\max}\). The validation confirms that the recipe is conservative yet effective in practice.

\section{Notation and Preliminaries}
\begin{table}[h]
\centering
\begin{tabular}{ll}
\hline
\textbf{Symbol} & \textbf{Description} \\
\hline
$\mathbf{G} \in \mathbb{R}^{N \times N}$ & Symmetric, entrywise nonnegative adjacency (undirected weighted graph) \\
$\mathbf{L} = \mathbf{D} - \mathbf{G}$ & Combinatorial graph Laplacian (PSD) \\
$\mathbf{D} = \operatorname{diag}(\mathbf{G}\one)$ & Degree matrix \\
$d_{\max}$ & Maximum degree $=\|\mathbf D\|_2=\max_i D_{ii}$ \\
$\tilde{\mathbf{R}}_s, \tilde{\mathbf{R}}_i$ & Signal and interference covariance matrices (PSD) \\
$\mathbf{R}_{\text{in}}$ & Interference-plus-noise (no Laplacian term) \\
$\mathbf{w} \in \mathbb{C}^N$ & Design vector (complex weights) \\
$\mathbf{p}_t \in \mathbb{R}_+^N$ & Power (or intensity) distribution at time $t$ \\
$\rho \in (0,1)$ & Propagation factor \\
$\mu > 0$ & Laplacian regularisation parameter \\
$\xi$ & Relative steady-state spreading \\
$\sigma^2$ & Noise variance \\
$\alpha$ & Diagonal loading parameter \\
$\Id[N]$ & $N \times N$ identity matrix (we write $\Id$ when the size is clear) \\
$\one$ & Vector of ones (length $N$) \\
$(\cdot)^H$ & Hermitian transpose \\
$\mathrm{tr}(\cdot)$ & Matrix trace \\
$\|\cdot\|_2$ & Vector 2-norm / matrix spectral norm (as appropriate) \\
$\lambda_{\max}(\cdot)$ & Maximum eigenvalue (spectral radius for PSD) \\
\hline
\end{tabular}
\end{table}

\paragraph{Standing assumptions.}
We assume $\mathbf G=\mathbf G^\top\!\ge 0$ (entrywise) so that $\mathbf L=\mathbf D-\mathbf G\succeq 0$, and $\rho\|\mathbf G\|_2<1$. We also assume $\one^\top\tilde{\mathbf R}_s\one>0$ (mild and standard), used once in Step~4 of the proof. Throughout, $\|\cdot\|_2$ on matrices denotes the spectral (operator) norm. We treat $\one$ as real, so $\one^\top=\one^{H}$.

\section{System Model}\label{sec:system-model}
\paragraph{Graph and diffusion.}
Let $\mathbf G\in\mathbb R^{N\times N}$ be a symmetric, entrywise nonnegative adjacency describing an undirected weighted graph on $N$ nodes. Define the degree matrix $\mathbf D=\operatorname{diag}(\mathbf G\one)$ and the combinatorial Laplacian $\mathbf L=\mathbf D-\mathbf G\succeq 0$. We assume the stability condition $\rho\|\mathbf G\|_2<1$ for a fixed $\rho\in(0,1)$. The elementwise power (or intensity) evolves according to
\begin{equation}
\mathbf p_{t+1}=\rho\,\mathbf G\,\mathbf p_t+(1-\rho)\,\mathbf p_0,
\qquad
\mathbf p_0\in\mathbb R_+^N,
\label{eq:model-diffusion}
\end{equation}
with steady state $\mathbf p_\infty=(\Id-\rho\mathbf G)^{-1}(1-\rho)\mathbf p_0$ whenever $\rho\|\mathbf G\|_2<1$ (assumed throughout).

\paragraph{Design and initial profile.}
In an array-processing application, the design stage produces a weight vector $\mathbf w^\star(\mu)$ by solving a Laplacian-regularised optimisation problem on the graph
(see Section 5.1). The resulting initial profile is $\mathbf p_0 = |\mathbf w^\star|^2$.

\paragraph{Performance metric and constants.}
We quantify spreading by
\(
\xi=\|\mathbf p_\infty-\mathbf p_0\|_2/\|\mathbf p_0\|_2.
\)
The bound in Theorem~\ref{thm:spreadBound} depends on
\[
C(\rho,\mathbf G)=\frac{(1-\rho)\rho}{1-\rho\|\mathbf G\|_2},\quad
d_{\max}=\|\mathbf D\|_2,\quad
\lambda_{\max}(\mathbf L),\quad
\lambda_{\max}(\tilde{\mathbf R}_s),\quad
\Lambda_{\mathrm{ref}}=\frac{\one^\top\mathbf R_{\mathrm{in}}\one}{\one^\top\tilde{\mathbf R}_s\one}.
\]
The practical design rule (Corollary~\ref{cor:design}) asserts that a target $\xi_{\mathrm{target}}$ is enforceable by choosing
\(
\mu \ge \frac{4\,N\,C(\rho,\mathbf G)^2\,\lambda_{\max}(\mathbf L)\,\Lambda_{\mathrm{ref}}\,\lambda_{\max}(\tilde{\mathbf R}_s)}
{(\xi_{\mathrm{target}}-C(\rho,\mathbf G)d_{\max})^2}
\)
provided $\xi_{\mathrm{target}}>C(\rho,\mathbf G)d_{\max}$.

\paragraph{Design and diffusion layers on the same graph.}
The graph $\mathbf G$ appears at two distinct layers. At \emph{design time}, the Laplacian penalty $\mu\,\mathbf w^H \mathbf L \mathbf w$ (with $\mathbf L = \mathbf D - \mathbf G$) encourages graph-smooth designs by penalising variations of $\mathbf w$ across edges. At \emph{run time}, after optimisation, the initial profile $\mathbf p_0 = |\mathbf w^\star|^2$ is injected into the same graph and evolves according to the diffusion model~\eqref{eq:model-diffusion}. Using a single graph for both layers is natural: the same neighbourhood structure that enforces similarity during design also governs the physical or logical coupling along which the profile may later spread.

\paragraph{Physical meaning of the variables.}
In an array-processing interpretation, $\mathbf w$ collects complex beamformer weights, $\mathbf p_0 = |\mathbf w^\star|^2$ describes power per antenna or element, $\mathbf R_{\mathrm{in}}$ models interference-plus-noise, and the Laplacian $\mathbf L$ encodes adjacency of elements in space. In a networked-sensing interpretation, $\mathbf p_t$ captures per-node power, traffic, or sensing load, and $\mathbf G$ represents communication or interference links, with $\mu$ controlling how strongly the design trades off performance against graph-smoothness. Our theoretical guarantees only rely on these structural roles (graph, Laplacian, covariance, and regularisation), which is why the resulting bounds transfer across such domains.

\section{Mathematical Framework}
\subsection{Optimisation Formulation}
In one application (array processing), $\tilde{\mathbf R}_s$ and $\tilde{\mathbf R}_i$ denote signal and interference covariances. We model interference-plus-noise \emph{without} the Laplacian term:
\begin{equation}
\mathbf{R}_{\text{in}} \;=\; \tilde{\mathbf{R}}_i \;+\; \sigma^2\Id[N] + \alpha\frac{\mathrm{tr}(\tilde{\mathbf R}_i)}{N}\Id[N].
\end{equation}
We design the weights via a penalised criterion that promotes graph-smoothness:
\begin{equation}
\min_{\mathbf{w}\neq \mathbf{0}} \ \frac{\mathbf{w}^H\mathbf{R}_{\text{in}}\mathbf{w}}{\mathbf{w}^H\tilde{\mathbf{R}}_s\mathbf{w}} \;+\; \mu\,\mathbf{w}^H\mathbf{L}\mathbf{w}.
\label{eq:penalized}
\end{equation}
Equivalently, using the constraint $\mathbf{w}^H\tilde{\mathbf{R}}_s\mathbf{w}=1$,
\begin{equation}
\min_{\mathbf{w}\neq \mathbf{0}} \ \mathbf{w}^H\mathbf{R}_{\text{in}}\mathbf{w} \;+\; \mu\,\mathbf{w}^H\mathbf{L}\mathbf{w}
\quad \text{s.t.} \quad \mathbf{w}^H\tilde{\mathbf{R}}_s\mathbf{w}=1.
\label{eq:constrained}
\end{equation}
Stationarity yields the generalised eigenproblem
\begin{equation}
(\mathbf R_{\mathrm{in}}+\mu\mathbf L)\,\mathbf w \;=\; \lambda\,\tilde{\mathbf R}_s\,\mathbf w,
\end{equation}
so $\mathbf w^\star$ is the generalised eigenvector corresponding to the smallest eigenvalue of the pair $(\mathbf R_{\mathrm{in}}+\mu\mathbf L,\tilde{\mathbf R}_s)$.

\subsection{Interference Control Analysis}
\label{sec:interferenceAnalysis}
We now analyse the diffusion model introduced in Section~\ref{sec:system-model}, see~\eqref{eq:model-diffusion}, where $\mathbf{p}_t \in \mathbb{R}_{+}^{N}$ collects the instantaneous power (or intensity) at each node and $\mathbf{p}_0 = |\mathbf{w}^\star|^{2}$ denotes the initial pattern produced by the design.

\begin{theorem}[Interference-spreading bound]
\label{thm:spreadBound}
Let $\mathbf{G} \in \mathbb{R}^{N\times N}$ be symmetric and entrywise nonnegative with $\rho\|\mathbf{G}\|_2<1$, and define $\mathbf{L}=\mathbf{D}-\mathbf{G}$ with $\mathbf{D}=\operatorname{diag}(\mathbf{G}\one)$. Let $\mathbf{w}^{\star}(\mu)$ solve \eqref{eq:penalized} (equivalently \eqref{eq:constrained}) and set $\mathbf{p}_0 = |\mathbf{w}^{\star}(\mu)|^{2}$. If $p_\infty$ denotes the steady state of \eqref{eq:model-diffusion} and 
$\xi = \|p_\infty - p_0\|_2/\|p_0\|_2$, then
\begin{equation}
\boxed{\;
\xi \;\le\; C(\rho,\mathbf G)\,\Big(\, d_{\max} \;+\;
\sqrt{ \tfrac{4\,N\,\lambda_{\max}(\mathbf L)\,\Lambda_{\mathrm{ref}}\,\lambda_{\max}(\tilde{\mathbf R}_s)}{\mu} }\,\Big),
\qquad
C(\rho,\mathbf G):=\frac{(1-\rho)\rho}{1-\rho\|\mathbf G\|_2}\; }
\label{eq:xiBound-correct}
\end{equation}
where $\displaystyle \Lambda_{\mathrm{ref}}:=\frac{\one^\top\mathbf R_{\mathrm{in}}\one}{\one^\top\tilde{\mathbf R}_s\one}$ and $d_{\max}=\|\mathbf D\|_2=\max_i D_{ii}$. Hence $\xi=O(\mu^{-1/2})$ as $\mu\to\infty$.
\end{theorem}

\begin{proof}
\textbf{Step 1: Rewrite $\xi$.}
The recursion \eqref{eq:model-diffusion} is affine-linear 
with fixed input $(1-\rho)\,p_0$, and its fixed point is
\[
\mathbf{p}_{\infty}=(\Id-\rho\mathbf{G})^{-1}(1-\rho)\mathbf{p}_0.
\]
Therefore
\[
\xi = \frac{\|[(\Id-\rho\mathbf{G})^{-1}(1-\rho)-\Id]\mathbf{p}_0\|_2}{\|\mathbf{p}_0\|_2}.
\]

\textbf{Step 2: Neumann series.}
Since $\rho\|\mathbf{G}\|_2<1$,
\[
(\Id-\rho\mathbf{G})^{-1} = \sum_{k=0}^{\infty}(\rho\mathbf{G})^{k}.
\]
Discard the $k=0$ term and use the triangle inequality to obtain
\begin{align}
\xi &\leq (1-\rho)\sum_{k=1}^{\infty}\rho^{k}\frac{\|\mathbf{G}^{k}\mathbf{p}_0\|_2}{\|\mathbf{p}_0\|_2}
\leq (1-\rho)\sum_{k=1}^{\infty}\rho^{k}\|\mathbf G\|_2^{k-1}\frac{\|\mathbf G\mathbf p_0\|_2}{\|\mathbf p_0\|_2}
\\
&= (1-\rho)\,\frac{\rho}{1-\rho\|\mathbf{G}\|_2}\cdot \frac{\|\mathbf{G}\mathbf{p}_0\|_2}{\|\mathbf{p}_0\|_2}.
\tag{A}
\end{align}

\textbf{Step 3: Relate $\mathbf{G}\mathbf{p}_0$ to $\mathbf{L}$.}
Write $\mathbf p_0=|\mathbf w^\star|^2$ and $\mathbf L=\mathbf D-\mathbf G$ with $\mathbf G=\mathbf G^\top\ge0$.
For any $\mathbf v$,
\begin{equation}\label{eq:Gsplit}
\|\mathbf G\mathbf v\|_2 \;\le\; d_{\max}\,\|\mathbf v\|_2 \;+\; \sqrt{\lambda_{\max}(\mathbf L)}\,\sqrt{\mathbf v^H\mathbf L\,\mathbf v}.
\end{equation}
Since $\mathbf G=\mathbf D-\mathbf L$, we have $\|\mathbf G\mathbf v\|_2\le \|\mathbf D\mathbf v\|_2+\|\mathbf L\mathbf v\|_2$, and by $\mathbf L$-inner-product arguments $\|\mathbf L\mathbf v\|_2\le\sqrt{\lambda_{\max}(\mathbf L)}\,\sqrt{\mathbf v^H\mathbf L\,\mathbf v}$.
Using the edge form of $\mathbf L$ and the inequality $\lvert |a|^2-|b|^2\rvert \le |a-b|\,\big(|a|+|b|\big)\le 2\|(a,b)\|_\infty\,|a-b|$,
for $\mathbf p_0=|\mathbf w^\star|^2$ we get
\begin{equation}\label{eq:p0Lp0}
\mathbf p_0^H\mathbf L\,\mathbf p_0 \;\le\; 4\,\|\mathbf w^\star\|_2^2\,\mathbf w^{\star H}\mathbf L\,\mathbf w^\star.
\end{equation}
Combining \eqref{eq:Gsplit}--\eqref{eq:p0Lp0} with $\mathbf v=\mathbf p_0$ and dividing by $\|\mathbf p_0\|_2$ yields
\begin{equation}\tag{B$'$}
\frac{\|\mathbf G\mathbf p_0\|_2}{\|\mathbf p_0\|_2}
\;\le\; d_{\max}
\;+\; \sqrt{4\,\lambda_{\max}(\mathbf L)}\,
\frac{\|\mathbf w^\star\|_2}{\|\mathbf p_0\|_2}\,
\sqrt{\mathbf w^{\star H}\mathbf L\,\mathbf w^\star}.
\end{equation}

\textbf{Step 4: Influence of $\mu$ (KKT with a reference vector).}
The first-order (KKT) conditions for \eqref{eq:constrained} give
\[
(\mathbf R_{\mathrm{in}}+\mu \mathbf L)\,\mathbf w^\star=\lambda^\star\,\tilde{\mathbf R}_s\,\mathbf w^\star, 
\qquad \mathbf w^{\star H}\tilde{\mathbf R}_s\mathbf w^\star=1,
\]
so $\lambda^\star=\mathbf w^{\star H}\mathbf R_{\mathrm{in}}\mathbf w^\star+\mu\,\mathbf w^{\star H}\mathbf L\mathbf w^\star\ge \mu\,\mathbf w^{\star H}\mathbf L\mathbf w^\star$.
For any feasible $\mathbf v\neq 0$ with $\mathbf v^H\tilde{\mathbf R}_s\mathbf v=1$,
\[
\lambda^\star\;=\;\min_{\mathbf w^H\tilde{\mathbf R}_s\mathbf w=1}\ \big(\mathbf w^H\mathbf R_{\mathrm{in}}\mathbf w+\mu\,\mathbf w^H\mathbf L\mathbf w\big)
\;\le\; \mathbf v^H\mathbf R_{\mathrm{in}}\mathbf v+\mu\,\mathbf v^H\mathbf L\mathbf v.
\]
Choosing $\mathbf v = \dfrac{\one}{\sqrt{\one^\top\tilde{\mathbf R}_s\one}}$ gives $\mathbf v^H\tilde{\mathbf R}_s\mathbf v=1$ and $\mathbf v^H\mathbf L\mathbf v=0$ (since $\mathbf L\one=\mathbf 0$), whence
\[
\lambda^\star\;\le\; \Lambda_{\mathrm{ref}}:=\frac{\one^\top\mathbf R_{\mathrm{in}}\one}{\one^\top\tilde{\mathbf R}_s\one},
\qquad\Rightarrow\qquad
\mathbf w^{\star H}\mathbf L\mathbf w^\star\;\le\;\frac{\Lambda_{\mathrm{ref}}}{\mu}.
\tag{C}
\]
\emph{More generally, for any feasible reference vector $\mathbf v$ with $\mathbf v^H\tilde{\mathbf R}_s\mathbf v=1$, we obtain}
\[
\lambda^\star\le \Lambda_{\mathrm{ref}}(\mathbf v):=\mathbf v^H\mathbf R_{\mathrm{in}}\mathbf v+\mu\,\mathbf v^H\mathbf L\mathbf v,
\quad\text{so}\quad
\mathbf w^{\star H}\mathbf L\mathbf w^\star\le \frac{\Lambda_{\mathrm{ref}}(\mathbf v)}{\mu}.
\]
\emph{If additionally $\mathbf v\in\ker\mathbf L$ (e.g., $\mathbf v\propto \one$), then $\Lambda_{\mathrm{ref}}(\mathbf v)$ is $\mu$-independent, yielding the cleanest constant.}

\textbf{Step 5: Combine.}
By the standard norm relation $\|w\|_4 \ge N^{-1/4}\|w\|_2$,
\[
\|\mathbf p_0\|_2 = \||\mathbf w^\star|^2\|_2 = \|\mathbf w^\star\|_4^2
\;\ge\; \frac{\|\mathbf w^\star\|_2^2}{\sqrt{N}},
\]
so
\[
\frac{\|\mathbf w^\star\|_2}{\|\mathbf p_0\|_2}
\;\le\; \frac{\sqrt{N}}{\|\mathbf w^\star\|_2}.
\]
Under $\mathbf w^{\star H}\tilde{\mathbf R}_s\mathbf w^\star=1$ we have
$1 \le \lambda_{\max}(\tilde{\mathbf R}_s)\,\|\mathbf w^\star\|_2^2$, hence
$\|\mathbf w^\star\|_2 \ge 1/\sqrt{\lambda_{\max}(\tilde{\mathbf R}_s)}$, so
\[
\frac{1}{\|\mathbf w^\star\|_2}
\;\le\; \sqrt{\lambda_{\max}(\tilde{\mathbf R}_s)},
\]
and therefore
\[
\frac{\|\mathbf w^\star\|_2}{\|\mathbf p_0\|_2}
\;\le\; \sqrt{N\,\lambda_{\max}(\tilde{\mathbf R}_s)}.
\]
Substituting (B$'$) and (C) into (A) yields
\[
\xi \;\le\; \frac{(1-\rho)\rho}{1-\rho\|\mathbf G\|_2}\;
\Big(\, d_{\max} \;+\; \sqrt{4\,N\,\lambda_{\max}(\mathbf L)\,\Lambda_{\mathrm{ref}}\,\lambda_{\max}(\tilde{\mathbf R}_s)/\mu}\,\Big),
\]
which is \eqref{eq:xiBound-correct} (with $d_{\max}=\|\mathbf D\|_2$).
\end{proof}

\begin{remark}[Euclidean normalisation variant]
If, instead of $\,\mathbf w^H\tilde{\mathbf R}_s\mathbf w=1\,$, we enforce the Euclidean constraint $\,\|\mathbf w\|_2=1\,$ in \eqref{eq:constrained}, the same argument yields
\[
\xi \;\le\; C(\rho,\mathbf G)\Big(d_{\max}+\sqrt{\tfrac{4\,\lambda_{\max}(\mathbf L)\,(\one^\top\mathbf R_{\mathrm{in}}\one)}{\mu}}\Big).
\]
Here $\,\mathbf w^{\star H}\mathbf L\mathbf w^\star\le \tfrac{1}{N}\tfrac{\one^\top\mathbf R_{\mathrm{in}}\one}{\mu}\,$ from the KKT step with $\,\mathbf v=\one/\sqrt{N}\,$, and $\,\|\mathbf p_0\|_2=\|\mathbf w^\star\|_4^2\ge N^{-1/2}\,$ under $\|\mathbf w^\star\|_2=1$, so the $\sqrt{N}$ and $1/\sqrt{N}$ factors cancel, removing both $N$ and $\lambda_{\max}(\tilde{\mathbf R}_s)$.
\end{remark}

\begin{remark}[Physical interpretation]
The bound decays as $1/\sqrt{\mu}$, so doubling $\mu$ reduces spreading by a factor $\sqrt{2}$. The additive degree term $d_{\max}$ quantifies baseline coupling: when degrees are large, diffusion pressure persists even for strongly smoothed designs.
\end{remark}

\begin{remark}[Bound tightness and normalisation choice]
The estimate leverages the triangle and Cauchy--Schwarz inequalities and is generally not tight. Tightness is approached when $\mathbf p_0$ aligns with the Perron vector of $\mathbf G$ (dominant eigenvector of a symmetric nonnegative matrix).
\textbf{If instead you enforce the Euclidean normalisation $\|\mathbf w\|_2=1$}
in \eqref{eq:constrained}, Step~4 still holds with the reference $\mathbf v=\one/\sqrt{N}$
and yields $\Lambda_{\mathrm{ref}}=(\one^\top\mathbf R_{\mathrm{in}}\one)/N$; moreover,
the factor $\sqrt{\lambda_{\max}(\tilde{\mathbf R}_s)}$ in \eqref{eq:xiBound-correct} disappears
\textbf{and the factor $\sqrt{N}$ cancels, so the $1/\sqrt{\mu}$ term depends only on
$\lambda_{\max}(\mathbf L)$ and $\one^\top\mathbf R_{\mathrm{in}}\one$.}
\end{remark}

\begin{corollary}[Design guideline]\label{cor:design}
Let $C(\rho,\mathbf G):=\frac{(1-\rho)\rho}{1-\rho\|\mathbf G\|_2}$ and $\Lambda_{\mathrm{ref}}=\frac{\one^\top\mathbf R_{\mathrm{in}}\one}{\one^\top\tilde{\mathbf R}_s\one}$. To guarantee a target $\xi\le \xi_{\mathrm{target}}$, it suffices to choose $\mu$ such that
\[
\xi_{\mathrm{target}} \;>\; C(\rho,\mathbf G)\,d_{\max}
\quad\text{and}\quad
\boxed{\;
\mu \;\ge\;
\frac{4\,N\,C(\rho,\mathbf G)^2\,\lambda_{\max}(\mathbf L)\,\Lambda_{\mathrm{ref}}\,\lambda_{\max}(\tilde{\mathbf R}_s)}
{\big(\,\xi_{\mathrm{target}}-C(\rho,\mathbf G)\,d_{\max}\,\big)^2}}.
\]
If $\xi_{\mathrm{target}}\le C(\rho,\mathbf G)\,d_{\max}$, the bound cannot enforce the target for any finite $\mu$.
\end{corollary}

\begin{remark}[Design rule under $\|\mathbf w\|_2=1$]
With the Euclidean constraint, a sufficient choice is
\[
\mu \;\ge\;
\frac{4\,C(\rho,\mathbf G)^2\,\lambda_{\max}(\mathbf L)\,(\one^\top\mathbf R_{\mathrm{in}}\one)}
{\big(\xi_{\mathrm{target}}-C(\rho,\mathbf G)d_{\max}\big)^2}.
\]
\end{remark}

\paragraph{Practical impact.}
Corollary~\ref{cor:design} links an interference budget $\xi_{\mathrm{target}}$ directly to a closed-form choice of $\mu$, removing trial-and-error in tuning Laplacian-regularised designs. The baseline term $C(\rho,\mathbf G)\,d_{\max}$ highlights when network degrees alone exceed the target, signalling the need to redesign the network/graph or reduce $\rho$.

\paragraph{Conservatism in practice.}
Since the theorem is worst-case, the rule is \emph{sufficient} by design. A two-step tuning procedure is: (i) compute $\mu$ from the corollary; (ii) verify the obtained $\xi$ on the instance and, if desired, reduce $\mu$ while maintaining the target with margin.

\section{Results and Discussion}\label{sec:results}

\textbf{Evaluation graph construction (reproducibility).}
We generate an undirected weighted graph $\mathbf G\in\mathbb R^{N\times N}$ from node locations $x_i\in[0,1]^2$ sampled i.i.d.\ uniformly with a fixed random seed, using $N=120$ nodes. For each node $i$, we connect it to its $k$ nearest neighbours in Euclidean distance (with $k=6$) and assign weights $G_{ij}=\exp(-d_{ij}/\sigma_i)$ for connected pairs, where $d_{ij}=\|x_i-x_j\|_2$ and $\sigma_i$ is the median distance from node $i$ to its $k$ neighbours. We then symmetrise the adjacency via $\mathbf G\leftarrow\max(\mathbf G,\mathbf G^\top)$, set $G_{ii}=0$, and normalise by the maximum row sum so that $d_{\max}=\max_i\sum_j G_{ij}\approx 1$. Finally, we select $\rho\in(0,1)$ with a fixed safety margin so that $\rho\|\mathbf G\|_2<1$ (in the implementation, a $10\%$ stability margin is used and $\rho$ is capped at $0.15$), ensuring that the steady state $\mathbf p_\infty=(\Id-\rho\mathbf G)^{-1}(1-\rho)\mathbf p_0$ and the spreading metric $\xi$ are well-defined. It should be noted that for visualisation, we normalise $\mathbf p_0$ to sum to one; the metric $\xi$ is scale-invariant, so this normalisation does not affect the reported values.

The signal and interference covariances follow the model of Section~5.1, and the constants $d_{\max}$, $C(\rho,G)$, $\lambda_{\max}(L)$, and $\Lambda_{\mathrm{ref}}$ are computed from this instance. We include compact visual evidence to contextualise the theory without relying on implementation details. Figure~\ref{fig:schematic} provides a high-level schematic of the pipeline and constants that enter the bound. Figure~\ref{fig:xi-vs-bound} plots spreading $\xi(\mu)$ against the theoretical upper bound as $\mu$ varies on a log grid for a representative instance. Figure~\ref{fig:rho-effect} shows how both the bound and measurements change with the propagation factor $\rho$ at a fixed~$\mu$.

For the instance shown, the computed constants were $d_{\max}\!\approx\!1$, $\|\mathbf G\|_2\!\in(0,1)$ by construction, $\lambda_{\max}(\mathbf L)$ on the order of unity, and $\Lambda_{\mathrm{ref}}$ and $\lambda_{\max}(\tilde{\mathbf R}_s)$ determined from the synthetic covariances; together with the chosen $\rho$ these yield a modest prefactor $C(\rho,\mathbf G)$ and a small floor $C(\rho,\mathbf G)d_{\max}$.

In Figure~\ref{fig:xi-vs-bound} (log--log axes), the empirical spreading $\xi(\mu)$ is plotted against the theoretical bound of Theorem~\ref{thm:spreadBound}. A dotted horizontal line marks the feasibility floor $C(\rho,\mathbf G)\,d_{\max}$ for the \emph{bound}. A dash--dot reference line of slope $-1/2$ highlights the predicted $O(\mu^{-1/2})$ decay. Over the shown $\mu$-range the $\mu^{-1/2}$ term dominates the constant $d_{\max}$, so the bound appears nearly linear on log--log axes. The upper bound (certificate) is conservative but respected across the sweep. A log--log least-squares fit of $\log\xi(\mu)$ versus $\log\mu$ over the mid-range yields an empirical slope close to $-0.5$, corroborating the predicted scaling. Across the sweep, the theoretical certificate consistently upper-bounded all measurements, as predicted and in line with the analysis.

Figure~\ref{fig:rho-effect} studies sensitivity to the propagation factor: at fixed $\mu=10$, both the bound and the measured $\xi$ increase monotonically with $\rho$, mirroring the prefactor $C(\rho,\mathbf G)$ and illustrating the role of dynamic coupling.

\begin{remark}[Bend point of the bound]
The transition between the $\mu^{-1/2}$ regime and the degree-driven floor occurs around
$\displaystyle \mu_\star \;=\; \frac{4\,N\,\lambda_{\max}(\mathbf L)\,\Lambda_{\mathrm{ref}}\,\lambda_{\max}(\tilde{\mathbf R}_s)}{d_{\max}^2}$.
For $\mu \ll \mu_\star$ the bound follows $O(\mu^{-1/2})$; for $\mu \gg \mu_\star$ it flattens to $C(\rho,\mathbf G)\,d_{\max}$.
\end{remark}

\begin{figure}[t]
\centering
\resizebox{\linewidth}{!}{
\begin{tikzpicture}[node distance=12mm, >=Latex, font=\small]
\tikzset{block/.style={draw, rounded corners, align=center, inner sep=4pt, outer sep=2pt},
thinblock/.style={draw, rounded corners, align=left, inner sep=3pt, outer sep=2pt}}
\node[block, text width=52mm] (design) {Design $\,\mathbf w^\star$\\[1mm]
$\displaystyle \min_{\mathbf w^H\tilde{\mathbf R}_s\mathbf w=1}
\ \mathbf w^H\mathbf R_{\mathrm{in}}\mathbf w + \mu\,\mathbf w^H\mathbf L\mathbf w$};
\node[block, right=18mm of design] (p0) {Initial pattern\\[1mm] $\mathbf p_0 = |\mathbf w^\star|^2$};
\node[block, right=18mm of p0, text width=44mm] (diff) {Graph diffusion\\[1mm]
$\displaystyle \mathbf p_{t+1}=\rho\,\mathbf G\,\mathbf p_t + (1-\rho)\,\mathbf p_0$\\[0.5mm]
(stable if $\,\rho\|\mathbf G\|_2<1$)};
\node[block, right=18mm of diff, text width=44mm] (pinf) {Steady state \& spreading\\[1mm]
$\displaystyle \mathbf p_\infty=(\Id-\rho\mathbf G)^{-1}(1-\rho)\mathbf p_0$\\[0.5mm]
$\displaystyle \xi=\frac{\|\mathbf p_\infty-\mathbf p_0\|_2}{\|\mathbf p_0\|_2}$};
\node[thinblock, above=12mm of diff, text width=64mm] (const) {\textbf{Instance constants}\\[0.5mm]
$C(\rho,\mathbf G)=\dfrac{(1-\rho)\rho}{1-\rho\|\mathbf G\|_2}$, \;
$d_{\max}=\|\mathbf D\|_2$, \;
$\lambda_{\max}(\mathbf L)$, \;
$\lambda_{\max}(\tilde{\mathbf R}_s)$, \;
$\Lambda_{\mathrm{ref}}=\dfrac{\one^\top\mathbf R_{\mathrm{in}}\one}{\one^\top\tilde{\mathbf R}_s\one}$};
\node[thinblock, below=12mm of diff, text width=76mm] (bound) {\textbf{Upper bound (Theorem~\ref{thm:spreadBound})}\\[0.5mm]
$\displaystyle \xi \le C(\rho,\mathbf G)\Big(d_{\max} + 
\sqrt{\tfrac{4\,N\,\lambda_{\max}(\mathbf L)\,\Lambda_{\mathrm{ref}}\,
\lambda_{\max}(\tilde{\mathbf R}_s)}{\mu}}\Big)$};
\draw[->] (design) -- (p0);
\draw[->] (p0) -- (diff);
\draw[->] (diff) -- (pinf);
\draw[->, dashed] (const) -- (diff);
\draw[->, dashed] (const) -- (bound);
\draw[->, dashed] (diff) -- (bound);
\node[align=left, anchor=north west, text width=0.98\linewidth] at ($(design.south west)+(0mm,-5mm)$) {$\mathbf R_{\mathrm{in}}=\tilde{\mathbf R}_i+\sigma^2\Id+\alpha\frac{\mathrm{tr}(\tilde{\mathbf R}_i)}{N}\Id,\quad
\mathbf L=\mathbf D-\mathbf G$.};
\end{tikzpicture}
}
\caption{\textbf{Conceptual pipeline: design and diffusion.}
The left block represents the Laplacian-regularised design stage, where the optimisation in~\eqref{eq:constrained} produces $\mathbf p_0=|\mathbf w^\star|^2$, while the right block represents the diffusion stage, where $\mathbf p_t$ evolves on the same graph according to~\eqref{eq:model-diffusion} until it reaches $\mathbf p_\infty$. Theorem~\ref{thm:spreadBound} controls the spreading $\xi$ via instance constants, with an explicit $1/\sqrt{\mu}$ decay.}
\label{fig:schematic}
\end{figure}

\begin{figure*}[t]
  \centering
  \includegraphics[width=0.95\textwidth]{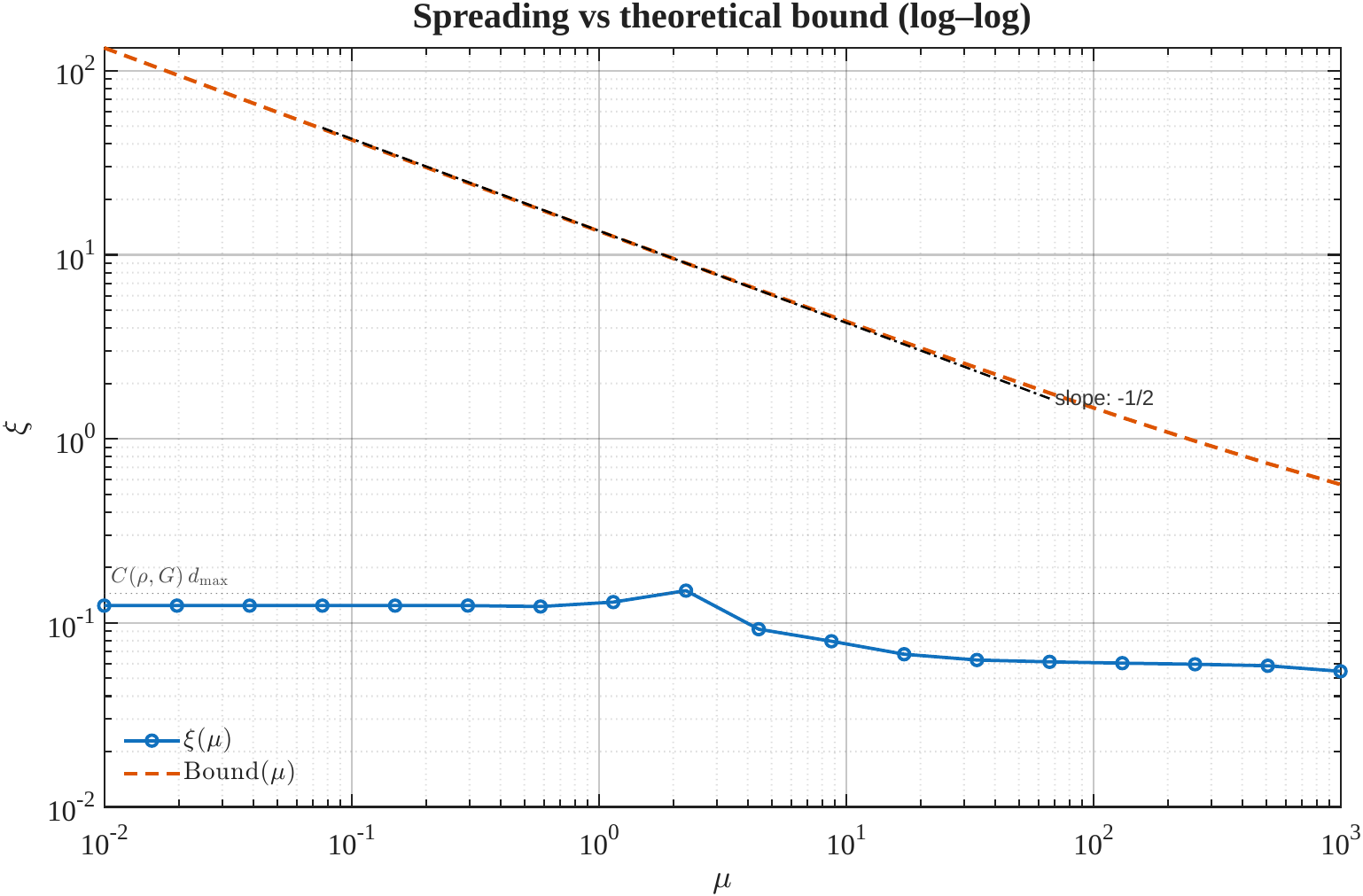}
  \caption{\textbf{Empirical spreading vs.\ theoretical bound (log--log).} Measured $\xi(\mu)$ (blue) decreases with $\mu$ and remains below the instance-wise upper bound from Theorem~\ref{thm:spreadBound} (orange dashed). A dash--dot reference line of slope $-1/2$ highlights the predicted $O(\mu^{-1/2})$ scaling. The dotted line marks the feasibility floor $C(\rho,\mathbf G)\,d_{\max}$ for the \emph{bound}. Over this $\mu$ range, the $\mu^{-1/2}$ term dominates, so the bound is nearly linear on log--log axes; it bends toward $C(\rho,\mathbf G)\,d_{\max}$ only for much larger $\mu$.}
  \label{fig:xi-vs-bound}
\end{figure*}

\begin{figure*}[t]
  \centering
  \includegraphics[width=0.80\textwidth]{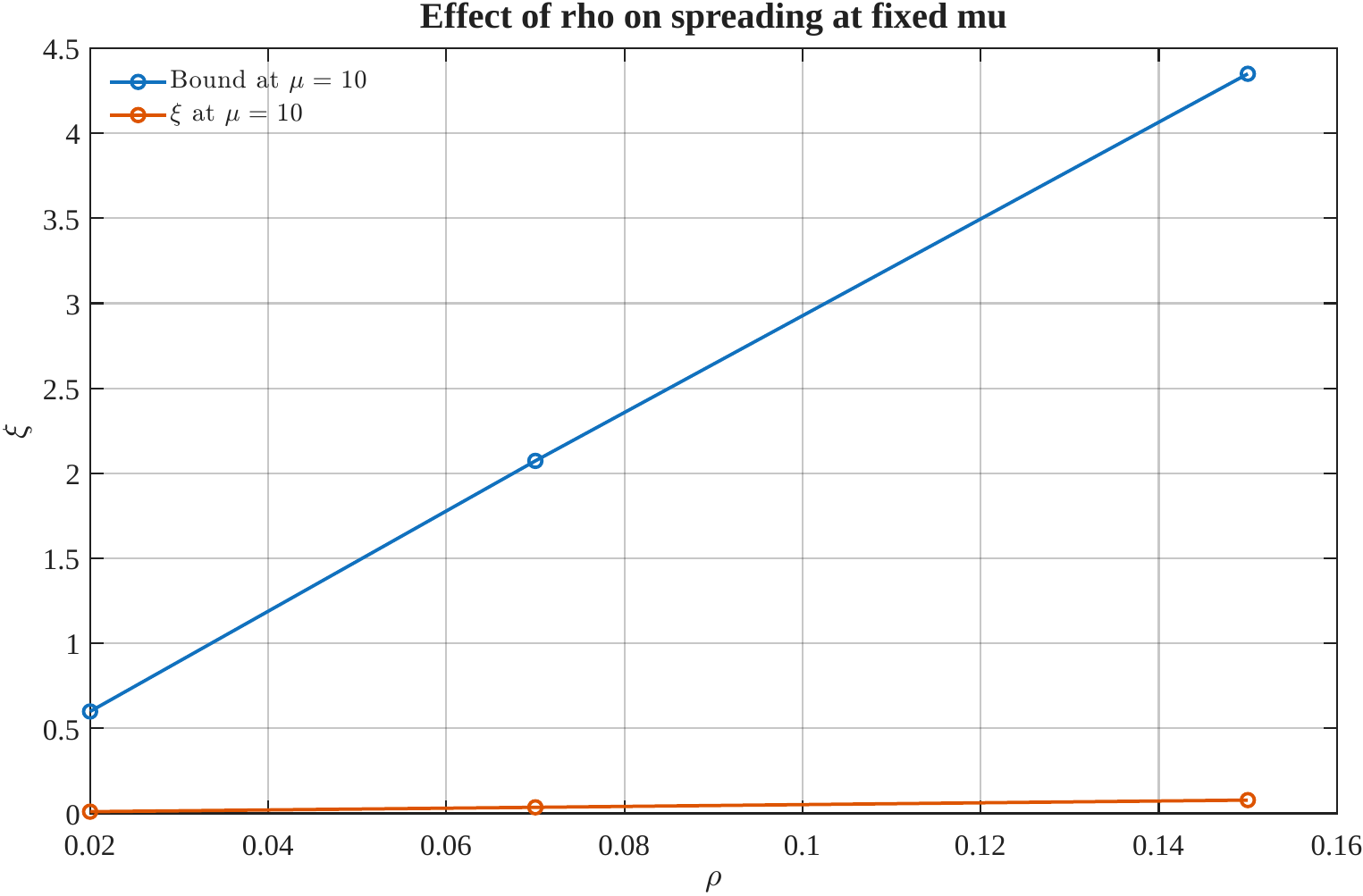}
  \caption{\textbf{Effect of propagation factor $\rho$ at fixed $\mu=10$.} Both the bound and the empirical $\xi$ increase monotonically with $\rho$, consistent with $C(\rho,\mathbf G)=\frac{(1-\rho)\rho}{1-\rho\|\mathbf G\|_2}$. For small $\rho$, the measured $\xi$ is close to zero because diffusion exerts only a weak push on the smoothed initial profile $\mathbf p_0$, while the worst-case bound varies according to $C(\rho,\mathbf G)$.}
  \label{fig:rho-effect}
\end{figure*}

\paragraph{Summary of contributions.}
We established a closed-form upper bound on spreading for a graph diffusion model driven by an initial profile $\mathbf p_0=|\mathbf w^\star|^2$, where $\mathbf w^\star$ solves a Laplacian-regularised design. The main result,
\[
\xi \;\le\; C(\rho,\mathbf G)\,\Big(\, d_{\max} \;+\; 
\sqrt{4\,N\,\lambda_{\max}(\mathbf L)\,\Lambda_{\mathrm{ref}}\,\lambda_{\max}(\tilde{\mathbf R}_s)/\mu}\,\Big),
\]
proves an explicit $O(\mu^{-1/2})$ decay of spreading and yields a practical design rule for selecting the regularisation strength $\mu$.

\paragraph{Interpretation of constants.}
The prefactor $C(\rho,\mathbf G)=\frac{(1-\rho)\rho}{1-\rho\|\mathbf G\|_2}$ captures how strongly the network accumulates interference over time; it diverges as $\rho\|\mathbf G\|_2\uparrow 1$.
The additive term $d_{\max}$ represents an \emph{irreducible} baseline due to node degrees: even when smoothing is strong, high-degree (highly connected) nodes contribute more strongly to diffusion pressure.
The term $\Lambda_{\mathrm{ref}}=\frac{\one^\top\mathbf R_{\mathrm{in}}\one}{\one^\top\tilde{\mathbf R}_s\one}$ is a computable ratio of problem covariances evaluated on the all-ones vector; using $\one$ is natural because $\mathbf L\one=\mathbf 0$, which eliminates the Laplacian penalty in the KKT upper bound.

\paragraph{Assumptions.}
We assume $\mathbf G$ is symmetric and entrywise nonnegative (undirected weighted graphs), implying $\mathbf L=\mathbf D-\mathbf G\succeq 0$, and that $\rho\|\mathbf G\|_2<1$ so the fixed point of the diffusion exists. The weights are computed and formulated with a Laplacian penalty \emph{outside} $\mathbf R_{\mathrm{in}}$, ensuring the $1/\sqrt{\mu}$ scaling. The analysis is deterministic and worst-case; it does not rely on distributional assumptions.

\paragraph{Tightness and conservatism.}
The certificate leverages standard inequalities (e.g., triangle and Cauchy–Schwarz) to prioritise robustness. A key source of looseness is Step~3, where we use $\|\mathbf p_0\|_2=\||\mathbf w^\star|^2\|_2=\|\mathbf w^\star\|_4^2\ge \|\mathbf w^\star\|_2^2/\sqrt N$ (norm monotonicity), which introduces an explicit $N$ factor and trades tightness for generality. Sharper control is possible under additional structure (e.g., bounded dynamic range or $\|\mathbf w^\star\|_\infty/\|\mathbf w^\star\|_2$ constraints), which would replace this step by a smaller factor and reduce the observed slack. It becomes tighter when $\mathbf p_0$ aligns with the Perron eigenvector of $\mathbf G$. Although constants may be improved for particular graph families (e.g., regular graphs or banded Toeplitz structures), the $1/\sqrt{\mu}$ scaling is fundamental to Laplacian smoothing.

\paragraph{Design implications.}
The corollary converts a spreading budget $\xi_{\mathrm{target}}$ into a closed-form choice of $\mu$, subject to the feasibility condition $\xi_{\mathrm{target}} > C(\rho,\mathbf G)\,d_{\max}$. If this condition is not met, the target cannot be \emph{certified} under the current graph and $\rho$; practitioners can then (i) reduce $\rho$, (ii) sparsify or reweight $\mathbf G$ to lower $d_{\max}$ and $\|\mathbf G\|_2$, or (iii) alter the geometry that induces $\mathbf G$.

In practice, one typically chooses $\mu$ so that the Laplacian term $\mu\,\mathbf w^H \mathbf L \mathbf w$ is of comparable order to the incoherence term $\mathbf w^H \mathbf R_{\mathrm{in}} \mathbf w$, rather than letting the design be dominated by regularisation alone. The design rule in Corollary~\ref{cor:design} should therefore be interpreted as providing a \emph{sufficient} value of $\mu$ to certify a given spreading budget, after which $\mu$ can be reduced while monitoring both $\xi$ and the primary performance metric (e.g., SINR or beam pattern quality). Very large $\mu$ do recover the asymptotic $1/\sqrt{\mu}$ scaling of the bound, but may yield overly smooth designs that are less useful in some applications; our recommended operating regime balances this trade-off.

\paragraph{Normalisation choices.}
We analysed the constrained form $\mathbf w^H\tilde{\mathbf R}_s\mathbf w=1$. If instead one enforces $\|\mathbf w\|_2=1$, the bound retains the same structure and the $\sqrt{\lambda_{\max}(\tilde{\mathbf R}_s)}$ factor vanishes, with $\Lambda_{\mathrm{ref}}$ replaced by $(\one^\top\mathbf R_{\mathrm{in}}\one)/N$. Both variants produce the same $1/\sqrt{\mu}$ law.

\paragraph{Future work.}
Natural extensions include: (i) directed or asymmetric graphs (requiring a suitable Laplacian generalisation and careful handling of non-normal $\mathbf G$); (ii) time-varying graphs $\mathbf G_t$ or propagation factors $\rho_t$; (iii) stochastic diffusion models with additive process noise; (iv) alternative smoothness priors (e.g., higher-order or total-variation graph penalties); and (v) data-driven graph learning jointly with the design stage. Empirical studies can compare graphs (line, grid, random geometric), and validate design rules for $\mu$. The proof technique applies beyond array processing to any setting where an initial profile diffuses on a graph and one can shape that profile via Laplacian regularisation (e.g., sensor networks and diffusion-based GNNs).

\section{Limitations}\label{sec:limitations}
Our guarantees take into account several modelling choices and inequalities. We group the main caveats by theme and indicate how a practitioner might mitigate each one. Taken together, these caveats delineate the envelope within which our framework is most informative. Rather than replacing application-specific metrics, the bound is best viewed as a lightweight certification layer that can be placed on top of existing design pipelines. Within the regime of linear, stable diffusion and Laplacian regularisation, it turns a handful of spectral descriptors of the graph and covariance structure into a guarantee on how far an engineered initial profile may drift under network dynamics. This separation between structural quantities and design choices makes the framework portable across domains that share a graph-based notion of neighbourhood but differ in their physical interpretation, while the limitations clarify which modelling ingredients must be revisited when moving beyond this regime.

\paragraph{Structural assumptions on the graph.}
We assume an undirected, entrywise nonnegative adjacency $\mathbf G=\mathbf G^\top\!\ge 0$, so that the combinatorial Laplacian $\mathbf L=\mathbf D-\mathbf G\succeq 0$ is well-defined. Directed or signed graphs require alternative Laplacian constructions and non-normal matrix tools; the present proof (which exploits $\mathbf L$-energy and symmetry) does not carry over directly. \emph{Mitigation:} for directed graphs, consider symmetrisation ($(\mathbf G+\mathbf G^\top)/2$) as a conservative surrogate, or develop bounds via the Hermitian part of $\mathbf G$ and pseudospectral analysis.

\paragraph{Stability and prefactor growth.}
The fixed point of the diffusion exists only when $\rho\|\mathbf G\|_2<1$. The bound's prefactor
$C(\rho,\mathbf G)=\frac{(1-\rho)\rho}{1-\rho\|\mathbf G\|_2}$ increases sharply as $\rho\|\mathbf G\|_2\uparrow 1$, which makes the certified bound large even if the measured $\xi$ remains moderate. \emph{Mitigation:} enforce a safety margin such as $1-\rho\|\mathbf G\|_2\gtrsim 0.1$ in design or simulation; report this margin alongside results.

\paragraph{Modelling choices that drive the $1/\sqrt{\mu}$ law.}
Two choices are essential: (i) the Laplacian penalty appears \emph{outside} $\mathbf R_{\mathrm{in}}$ in the design objective, and (ii) we normalise by $\mathbf w^H\tilde{\mathbf R}_s\mathbf w=1$. Absorbing $\mu\mathbf L$ into $\mathbf R_{\mathrm{in}}$ alters the $\mu$-dependence and can destroy the $1/\sqrt{\mu}$ scaling; switching to $\|\mathbf w\|_2=1$ removes the factor $\sqrt{\lambda_{\max}(\tilde{\mathbf R}_s)}$ but replaces $\Lambda_{\mathrm{ref}}$ with $(\one^\top\mathbf R_{\mathrm{in}}\one)/N$. \emph{Mitigation:} choose the normalisation that matches your pipeline and apply the corresponding corollary; avoid folding $\mu\mathbf L$ into $\mathbf R_{\mathrm{in}}$ if you wish to retain the scaling law.

\paragraph{Sources of conservatism.}
The derivation uses triangle and Cauchy--Schwarz inequalities, replaces vector-dependent quantities by spectral radii, and bounds $\|\mathbf p_0\|_2=\||\mathbf w^\star|^2\|_2$ via $\|\mathbf w^\star\|_4^2\ge \|\mathbf w^\star\|_2^2/\sqrt N$, which introduces an explicit $N$ factor. Consequently the bound is generally not tight; it approaches tightness only for special alignments. In other words, tightness depends on instance structure and can be strongest under favourable alignments (e.g., $\mathbf p_0$ aligned with the Perron eigenvector of $\mathbf G$). \emph{Mitigation:} where available, impose mild dynamic-range controls on $|\mathbf w^\star|$ (e.g., $\|\mathbf w^\star\|_\infty/\|\mathbf w^\star\|_2$) to sharpen Step~3; report both the measured $\xi$ and the certificate to contextualise the safety margin.

\paragraph{Irreducible baseline from degrees.}
The additive term $d_{\max}$ is structural: if $C(\rho,\mathbf G)\,d_{\max}\ge \xi_{\mathrm{target}}$, no finite $\mu$ can satisfy the target under this model. \emph{Mitigation:} reduce $\rho$, sparsify or reweight $\mathbf G$ to lower $d_{\max}$ and $\|\mathbf G\|_2$, or redesign the physical coupling that induces $\mathbf G$.

\paragraph{Scope of the dynamic model.}
We analyse a deterministic, time-invariant linear recursion with no process noise or nonlinearities. Time-varying graphs ($\mathbf G_t$), stochastic disturbances, saturation/clipping, or multiplicative effects are outside scope. \emph{Mitigation:} for slowly varying $\mathbf G_t$, apply the bound to a worst-case surrogate (e.g., $\sup_t\|\mathbf G_t\|_2$); for noisy settings, one can extend the recursion to include additive noise and bound second moments with similar techniques. Our diffusion recursion should thus be viewed as a linear surrogate for mixing or message-passing dynamics, intended to provide a transferable robustness certificate across arrays, wireless networks, and graph-signal or GNN architectures, rather than as an exact physical propagation law for any single system.

\paragraph{Instance dependence and graph-family sensitivity.}
The constants $C(\rho,\mathbf G)$, $d_{\max}$, $\lambda_{\max}(\mathbf L)$, $\Lambda_{\mathrm{ref}}$, and $\lambda_{\max}(\tilde{\mathbf R}_s)$ are instance-specific; they control practical tightness and vary across graph families. \emph{Mitigation:} report these constants for each experiment and include a short sensitivity sweep (regular, random geometric, scale-free) to show how the certificate behaves beyond a single instance. In this paper, we focus on a single representative synthetic instance to illustrate these dependencies, and we view a broad empirical comparison across multiple random graph families (e.g., Poisson and power-law degree distributions) as an important direction for future work rather than part of the present contribution.

\paragraph{Numerical considerations.}
While the theory is spectral, numerical eigensolvers can be sensitive to asymmetry and conditioning. \emph{Mitigation:} symmetrise matrices in computations, add a tiny SPD stabiliser to $\tilde{\mathbf R}_s$ when solving the generalised eigenproblem, and verify $\rho\|\mathbf G\|_2<1$ with a reproducible margin before forming $(\Id-\rho\mathbf G)^{-1}$.

\section{Conclusion}
Under Laplacian regularisation, we established a provable bound on spreading in graph-structured systems. Starting from the diffusion recursion $\mathbf{p}_{t+1}=\rho\mathbf{G}\mathbf{p}_t+(1-\rho)\mathbf{p}_0$, our analysis yields
\[
\xi \;\le\; C(\rho,\mathbf G)\,\Big(\, d_{\max} \;+\; 
\sqrt{4\,N\,\lambda_{\max}(\mathbf L)\,\Lambda_{\mathrm{ref}}\,\lambda_{\max}(\tilde{\mathbf R}_s)/\mu}\,\Big),
\]
which cleanly separates a degree-driven baseline from a design-controlled term that decays as $1/\sqrt{\mu}$. This gives a one-shot rule for choosing $\mu$ and clarifies when tighter targets are infeasible (when $\xi_{\mathrm{target}} \le C(\rho,\mathbf G)d_{\max}$).

Beyond the bound itself, the results identify two actionable levers: the propagation factor $\rho$ and the graph spectrum through $d_{\max}$ and $\|\mathbf G\|_2$. If the feasibility condition fails, increasing $\mu$ cannot help; instead one must reduce $\rho$ or modify the topology/weights that define $\mathbf G$.

The proof template—steady-state expansion, Laplacian energy control, and a KKT reference-vector argument—extends naturally to time-varying or learned graphs and to alternative smoothness penalties. An important next step is sharpening constants (e.g., via distributional assumptions or dynamic-range controls) while preserving the fundamental $1/\sqrt{\mu}$ law.

\paragraph{Outlook.}
The bound is a practical certification tool for any workflow that (i) engineers an initial graph-smooth profile and (ii) then lets the network dynamics run. Representative use cases include:
\begin{itemize}
\item \textit{Wireless/arrays and RIS/IRS.} Limiting pattern leakage or energy bleed across elements after deployment.
\item \textit{Cooperative sensing and distributed control.} Constraining how quickly local actions spread in consensus/averaging loops.
\item \textit{Graph ML.} Setting a safe regularisation scale to control over-smoothing in GNNs or label-propagation schemes.
\item \textit{Infrastructure networks.} Managing diffusion-like processes such as load balancing in microgrids or congestion propagation in transportation graphs and associated networked systems.
\end{itemize}
Operationally, the corollary turns a target budget $\xi_{\mathrm{target}}$ into a closed-form $\mu$—a drop-in replacement for sweep-based tuning—and provides a quick infeasibility test when network coupling is too strong. Future directions include directed/signed graphs (non-normal $\mathbf G$), stochastic disturbances, learned or time-varying topologies, and application-specific constants for canonical graph families.

\paragraph{}
\emph{Data Availability Statement:} The data supporting the findings of this study are available from the corresponding author upon reasonable request.

\emph{Conflict of Interest Statement:} No conflict of interest has been declared by the author.

\emph{Funding Information:} None.

\emph{Author Contribution:}
A. R.: Conceptualisation; data curation; formal analysis; investigation; methodology; software; validation; visualisation; writing original draft; writing review and editing.


\begin{thebibliography}{99}

\bibitem{Wu2021IRS}
Q.~Wu, S.~Zhang, B.~Zheng, C.~You, and R.~Zhang,
``Intelligent reflecting surface-aided wireless communications: A tutorial,''
\emph{IEEE Trans. Commun.},
vol.~69, no.~5, pp.~3313--3351, May 2021.

\bibitem{Gama2020Graphs}
F.~Gama, E.~Isufi, G.~Leus, and A.~Ribeiro,
``Graphs, convolutions, and neural networks: From graph filters to graph neural networks,''
\emph{IEEE Signal Process. Mag.},
vol.~37, no.~6, pp.~128--138, Nov. 2020.

\bibitem{Dong2020GSPMachineLearning}
X.~Dong, D.~Thanou, L.~Toni, M.~Bronstein, and P.~Frossard,
``Graph signal processing for machine learning: A review and new perspectives,''
\emph{IEEE Signal Process. Mag.},
vol.~37, no.~6, pp.~117--127, Nov. 2020.

\bibitem{Huang2023Robust}
Y.~Huang, H.~Fu, S.~A. Vorobyov, and Z.-Q. Luo,
``Robust adaptive beamforming via worst-case SINR maximization with nonconvex uncertainty sets,''
\emph{IEEE Trans. Signal Process.},
vol.~71, pp.~218--232, 2023.

\bibitem{Shuman2013Emerging}
D.~I. Shuman, S.~K. Narang, P.~Frossard, A.~Ortega, and P.~Vandergheynst,
``The emerging field of signal processing on graphs: Extending high-dimensional data analysis to networks and other irregular domains,''
\emph{IEEE Signal Process. Mag.},
vol.~30, no.~3, pp.~83--98, May 2013.

\bibitem{Sandryhaila2013Discrete}
A.~Sandryhaila and J.~M.~F. Moura,
``Discrete signal processing on graphs,''
\emph{IEEE Trans. Signal Process.},
vol.~61, no.~7, pp.~1644--1656, Apr. 2013.

\bibitem{Ortega2018GSP}
A.~Ortega, P.~Frossard, J.~Kovačević, J.~M.~F. Moura, and P.~Vandergheynst,
``Graph signal processing: Overview, challenges, and applications,''
\emph{Proc. IEEE},
vol.~106, no.~5, pp.~808--828, May 2018.

\bibitem{NaderiAlizadeh2023Learning}
N.~NaderiAlizadeh, M.~Eisen, and A.~Ribeiro,
``Learning resilient radio resource management policies with graph neural networks,''
\emph{IEEE Trans. Signal Process.},
vol.~71, pp.~995--1009, 2023.

\bibitem{Wang2024ENGNN}
Y.~Wang, Y.~Li, Q.~Shi, and Y.-C. Wu,
``ENGNN: A general edge-update empowered GNN architecture for radio resource management in wireless networks,''
\emph{IEEE Trans. Wireless Commun.},
vol.~23, no.~6, pp.~5330--5344, Jun. 2024.

\bibitem{Chung1997Spectral}
F.~R.~K. Chung,
\emph{Spectral Graph Theory},
CBMS Regional Conference Series in Mathematics, vol.~92.
American Mathematical Society, 1997.

\bibitem{Hammond2011Wavelets}
D.~K. Hammond, P.~Vandergheynst, and R.~Gribonval,
``Wavelets on graphs via spectral graph theory,''
\emph{Appl. Comput. Harmon. Anal.},
vol.~30, no.~2, pp.~129--150, Mar. 2011.

\bibitem{Leus2023GSPHistory}
G.~Leus, A.~G. Marques, J.~M.~F. Moura, A.~Ortega, and D.~I. Shuman,
``Graph signal processing: History, development, impact, and outlook,''
\emph{IEEE Signal Process. Mag.},
vol.~40, no.~4, pp.~49--60, Jun. 2023.

\bibitem{Ruiz2021Graphon}
L.~Ruiz, L.~F.~O. Chamon, and A.~Ribeiro,
``Graphon signal processing,''
\emph{IEEE Trans. Signal Process.},
vol.~69, pp.~4961--4976, 2021.

\bibitem{Feng2022DataReconstruction}
J.~Feng, F.~Chen, and H.~Chen,
``Data reconstruction coverage based on graph signal processing for wireless sensor networks,''
\emph{IEEE Wireless Commun. Lett.},
vol.~11, no.~1, pp.~48--52, Jan. 2022.

\bibitem{Yang2017InterferenceGraph}
J.~Yang, S.~C. Draper, and R.~Nowak,
``Learning the interference graph of a wireless network,''
\emph{IEEE Trans. Signal Inf. Process. Netw.},
vol.~3, no.~3, pp.~631--646, Sep. 2017.

\bibitem{Gama2019ConvolutionalNN}
F.~Gama, A.~G. Marques, G.~Leus, and A.~Ribeiro,
``Convolutional neural network architectures for signals supported on graphs,''
\emph{IEEE Trans. Signal Process.},
vol.~67, no.~4, pp.~1034--1049, Feb. 2019.

\bibitem{Thanou2017HeatDiffusion}
D.~Thanou, X.~Dong, D.~Kressner, and P.~Frossard,
``Learning heat diffusion graphs,''
\emph{IEEE Trans. Signal Inf. Process. Netw.},
vol.~3, no.~3, pp.~484--499, Sep. 2017.

\bibitem{Kar2009Distributed}
S.~Kar and J.~M.~F. Moura,
``Distributed consensus algorithms in sensor networks with imperfect communication: Link failures and channel noise,''
\emph{IEEE Trans. Signal Process.},
vol.~57, no.~1, pp.~355--369, Jan. 2009.

\bibitem{Yang2025DifferentiallyPrivate}
Z.~Yang, W.~He, and S.~Yang,
``Differentially private distributed optimization over time-varying unbalanced networks with linear convergence rates,''
\emph{IEEE Trans. Signal Process.},
vol.~73, pp.~1138--1152, 2025.

\end{thebibliography}
\end{document}